\documentclass[times, 10pt,twocolumn]{article} 
\usepackage[dvips]{graphicx}

\usepackage{array}

\usepackage[style=base]{caption}
\usepackage{subcaption}

\usepackage{graphics}
\usepackage[english]{babel}
\usepackage{amsthm}
\theoremstyle{definition}

\theoremstyle{remark}

\usepackage[noend]{algpseudocode}
\usepackage{algorithm}
\usepackage{algorithmicx}
\usepackage{listings} 
\usepackage{xfrac}

\makeatletter
\def\BState{\State\hskip-\ALG@thistlm}
\makeatother

\usepackage[latin1]{inputenc}
\usepackage{tikz}
\usepackage{verbatim}
\usetikzlibrary{arrows,positioning,shapes.geometric,intersections}

\usepackage{url}
\usepackage{chngpage}
\usepackage[hidelinks]{hyperref}
\usepackage{csquotes}
\usepackage{mathptmx}
\usepackage{rotating}
\usepackage{color, soul}
\usepackage{alltt}
\usepackage{upquote}
\usepackage{amsmath,amsthm,amssymb,amsfonts}
\usepackage[english]{babel}
\usepackage{bm}
\graphicspath{ {images/} }
\newtheorem{theorem}{Theorem}

\def\therefore{
\leavevmode
\lower0.1ex\hbox{$\bullet$}
\kern-0.2em\raise0.7ex\hbox{$\bullet$}
\kern-0.2em\lower0.2ex\hbox{$\bullet$}
\thinspace}

\makeatletter
\renewenvironment{proof}[1][\proofname]{\par
  \pushQED{\qed}%
  \normalfont \topsep6\p@\@plus6\p@\relax
  \trivlist
  \item[\hskip\labelsep
        \itshape
    #1\@addpunct{.}]\mbox{}\\*
}{%
  \popQED\endtrivlist\@endpefalse
}
\makeatother

\usepackage{multirow}
\usepackage{array}
\newcolumntype{L}[1]{>{\raggedright\let\newline\\\arraybackslash\hspace{0pt}}m{#1}}
\newcolumntype{C}[1]{>{\centering\let\newline\\\arraybackslash\hspace{0pt}}m{#1}}
\newcolumntype{R}[1]{>{\raggedleft\let\newline\\\arraybackslash\hspace{0pt}}m{#1}}

\usepackage{enumitem}

\makeatletter
\renewenvironment{proof}[1][\proofname]{\par
  \pushQED{\qed}%
  \normalfont \topsep6\p@\@plus6\p@\relax
  \trivlist
  \item[\hskip\labelsep
        \itshape
    #1\@addpunct{.}]\mbox{}\\*
}{%
  \popQED\endtrivlist\@endpefalse
}
\makeatother

\begin{document}

\title{Privacy Preserving Count Statistics}

\author{Lu Yu, Oluwakemi Hambolu, Yu Fu, Jon Oakley, and Richard R. Brooks\\
the Holcombe Department of Electrical and Computer Engineering\\ 
Clemson University, Clemson, SC, 29630\\
lyu,ohambol,fu2,joakley,rrb@g.clemson.edu\\
}

\maketitle
\thispagestyle{empty}

\begin{abstract}
The ability to preserve user privacy and anonymity is important. One of the safest ways to maintain privacy is to avoid storing personally identifiable information (PII), which poses a challenge for maintaining useful user statistics.  Probabilistic counting has been used to find the cardinality of a multiset when precise counting is too resource intensive.  In this paper, probabilistic counting is used as an anonymization technique that provides a reliable estimate of the number of unique users. We extend previous work in probabilistic counting by considering its use for preserving user anonymity, developing application guidelines and including hash collisions in the estimate.  Our work complements previous method by attempting to explore the causes of the deviation of uncorrected estimate from the real value. The experimental results show that if the proper register size is used, collision compensation provides estimates are as good as, if not better than, the original probabilistic counting.  We develop a new anonymity metric to precisely quantify the degree of anonymity the algorithm provides. 
\end{abstract}

\section{Introduction}
Privacy preserving data mining, also known as statistical disclosure control, inference control or private data analysis~\cite{Dwork:2008:DPS:1791834.1791836, Barak:2007:PAC:1265530.1265569} seeks to protect statistical data so they can be publicly released and mined while preserving privacy~\cite{series-ads-Domingo-Ferrer08}.  Related techniques are widely used in domains like medical informatics or opinion polling, where patients feel uncomfortable about divulging information on issues like their frequency of drug use, or HIV status~\cite{Esponda:2008:EIN:2209130.2209172}. 

Storage of personally identifiable information (PII) like name, social security number (SSN), IP or MAC addresses can be easily used to compromise users' privacy.  Consider the recent Ashley Madison data disclosure~\cite{Ashley}.  So the best way to preserve user privacy is to not keep any PII.  Methods like this are effective in foiling attackers, but make the task of maintaining user statistics very challenging.  This paper shows how to maintain an accurate estimate of the number of unique users while protecting the privacy of each individual user by using statistical counting. 

The work presented here builds on \textit{probabilistic counting} introduced by Flajolet and Martin~\cite{Flajolet:Martin}, which estimates the number of distinct records/users (cardinality) without keeping the records.  It is based on statistical analysis of the bit positions of hashed records.  The position of the least significant bit set in each hashed record is stored in a register called a \textit{BITMAP}.  From the bit position of the lowest bit that is not set on the \textit{BITMAP}, we get an unbiased estimate of the cardinality of the user set.  For example, the input records could be the users' SSL certificates. 

This paper makes three main contributions.  
\begin{itemize}
\item Probabilistic counting has been used for estimating the cardinality of a multiset when it is unrealistic (too resource-intensive to store or count the elements) to solve precisely~\cite{mcmullen2015probably}.  To the best of our knowledge, this is the first time the algorithm is used to preserve privacy and anonymity.  
\item We propose collision-included probabilistic counting (CIPC), which includes the hash collisions in the estimate. The Birthday paradox~\cite{Birthday} is used to determine the number of collisions.  The results of our experiments show that adding collisions to the uncorrected estimate of the number of users, gives a more accurate count than using a constant correction factor~\cite{Flajolet:Martin}. This also verifies that hash collision is a major cause of bias for probabilistic counting. 
\item We also provide an anonymity metric to measure the anonymity the proposed algorithm provides. 
\end{itemize}

This paper was inspired by problems we faced when maintaining user privacy while collecting usage statistics for a censorship circumvention tool~\cite{Hambolu:Thesis:2014}. The tool developed is used by dissidents and journalists, working in politically challenging regions in West Africa, to circumvent DNS and IP address blocking by leveraging technologies developed for use by criminal botnet enterprises~\cite{DBLP:conf/malware/2008}. 
To quantify the number of unique users of the system, a secure method was needed to keep track of the statistics.  To the best of our knowledge, this is the first time probabilistic counting is applied to safeguarding user privacy .  
 
The rest of the paper is organized as follows.  In section~\ref{sec:background} we present background on probabilistic counting and the hash function used in the proposed algorithm.  We present the collision-included probabilistic counting (CIPC) algorithm in section~\ref{sec:cipc}.  Experimental results are given along with the recommendation for selecting the proper register/\textit{BITMAP} size.  We conclude the paper and point out future work for this research in section~\ref{sec:conclusion}.

\section{Background}\label{sec:background}
This section covers previous probabilistic counting work.

Probabilistic counting~\cite{Flajolet:Martin} estimates the number of distinct elements (cardinality) in a large collection of data that contain duplicates. It requires only a small amount of storage and few operations per element.  The algorithm works as follows:
\begin{description}
\item[Inputs: ] A multiset of records $\chi'$; a register/\textit{BITMAP} $b$ of  length $L$ initialized to $zero$ (see Figure~\ref{Fig:bitmap}).
\item[Output: ] The estimate of the number of distinct records, denoted by $M$.
\item[Step 1] A hash function $h(x)$ that maps each element of $\chi'$ to an integer. This set of integers is uniformly distributed over the range $[0,1,\cdots, 2^{L}-1]$. 
\item[Step 2] For each record $x\in \chi'$, let $\text{LSSB}(h(x))$ be the position of the least significant bit that is set in the binary representation of $h(x)$,  $0\leq \text{LSSB}(h(x))\leq L-1$.  Set the corresponding \textit{BITMAP} position to $1$.  That is, if $\text{LSSB}(h(x))=i$, then $b[i]=1$.  For example, $b[4]$ will be set to $1$ if $h(x) = \cdots011010000$.
\item[Step 3] Let $k$ be the position of the rightmost zero in \textit{BITMAP} $b$.  The estimated value of the number of unique records in $\chi'$ is given by 
\begin{equation}\label{eq:varphi}
\tilde{M}_{PC} = \frac{1}{\varphi}\cdot 2^k 
\end{equation}
where $\varphi=0.77351\cdots$ is the correction factor.
\end{description}

The algorithm is based on statistical observations concerning the bits of hashed values.  Let $\chi$ denote the set of distinct records in $\chi'$.  Since $h(x)$ is uniformly distributed over $[0,1,\cdots,2^{L}-1]$, for each$x\in \chi$, $h(x) =\cdots1$ with probability $1/2$, $h(x) =\cdots10$ with probability $1/4$, $h(x)=\cdots 100$ with probability $1/8$ and so on.  This can be generalized as the sequence  $\cdots1\! \underbrace{0\cdots0\,}_\text{$k-1$}$ occurring with probability of ${1}/{2^{k}}$~\cite{Flajolet:Martin}.  So if we randomly select $x\in \chi$, the probability of $LSSB(x) = k$ being set to $1$ is ${1}/{2^{k}}$.  The \textit{BITMAP} only depends on the least significant set bits (LSSBs) of distinct hashed values and not on the frequency of the values.  If a record occurs more than once in multiset $\chi'$, it will be counted only once. 

\begin{figure}
\begin{center}
\includegraphics[width=.4\textwidth]{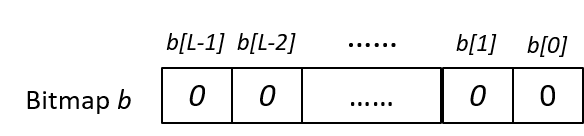}
\caption{Bitmap $b$ initialized to zero.}
\label{Fig:bitmap}
\end{center}
\end{figure} 

Let $M$ be the number of distinct records in $\chi'$, i.e., $M=|\chi|$, it is expected that $b[0]$ will be set about $M/2$ times, $b[1]$ will be set approximately $M/4$ times and so on.  At the end of the execution, it is very likely that $b[i]=0$ for $i \gg log_2M$ and $b[i]=1$ for $i \ll log_2M$ ~\cite{Flajolet:Martin}.  An example is illustrated in Figure~\ref{Fig:probmapping}.  Assuming there are $M =20000$ distinct records, about $10000$ records will have their LSSB be bit $0$, that is, about $M/2$ of the hash values are mapped to $b[0]$; about $5000$ records are mapped to bit position $1$, i.e., approximately $M/4$ records are mapped to $b[1]$, and so on~\cite{Gama:2010:KDD:1855075}.

\begin{figure}
\begin{center}
\includegraphics[width=.5\textwidth]{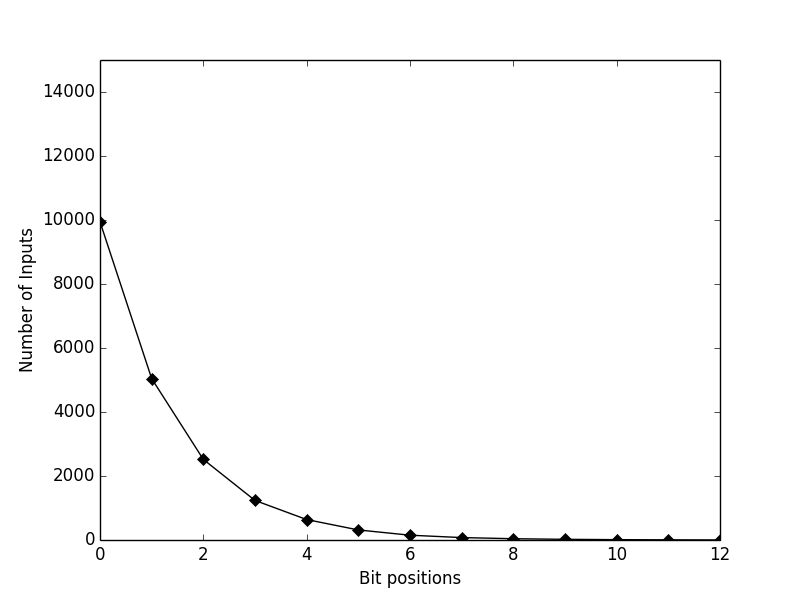}
\caption{The mapping of $20000$ distinct records to their bit position in the \textit{BITMAP}}
\label{Fig:probmapping}
\end{center}
\end{figure}

It was proposed in~\cite{Flajolet:Martin} to use the position of the rightmost zero in \textit{BITMAP} as an indicator of $\log_2 M$, based on the assumption that hash values will be uniformly distributed over $\left[0,1,\cdots ,2^L - 1\right]$.  For example, if the \textit{BITMAP} is $[\cdots1010 \underbrace{111111}_\text{$k-1=6$}]$, the cardinality $M$, is estimated by
\begin{equation*}
M\approx 2^k=2^6
\end{equation*}
It is shown in~\cite{Flajolet:Martin, Finch}, $2^k$ is a biased estimate of $M$ and a \textit{correction factor} $\varphi\approx 0.77351$ can be used as a simple correction.  The unbiased estimate of cardinality $M$ is given by
\begin{equation*}
\tilde{M}_{PC} = \frac{1}{\varphi}\cdot 2^k 
\end{equation*}

\section{Collision-Included Probabilistic Counting (CIPC)}\label{sec:cipc}
The plot of $\tilde{M}=2^k$ is in Figure~\ref{Fig:leveloff}, where $k$ is the position of the rightmost zero in the \textit{BITMAP}.  It is not hard to see that $2^k$ gives an underestimate of $M$ for the most part.  And the gap between $2^k$ and $M$ gets larger as $M$ increases.  According to~\eqref{eq:varphi}, probabilistic counting uses a correction factor $\varphi=0.77351\cdots$ to compensate for the deviation of using $2^k$ as an indicator of $M$.  Although rigorous calculation of $\varphi$ is provided in~\cite{Flajolet:Martin}, the cause of the deviation is not discussed.  In this section, we 
consider the effects of collisions in hashing and how including the collisions in hashing will improve the estimate.  The experimentation results show that our approach produces an estimate at least as good as using $\varphi$, or even better under certain circumstances.  Moreover, our approach provides some insight into the possible cause of the estimation bias.

\subsection{Collision Estimate using Birthday Paradox}
Even with uniform hash functions, it is inevitable that more than one record will be mapped to the same hash value, which is known as a \textit{collision}.  Although the \textit{BITMAP} is set to be big enough to hold all the records (i.e., $2^L>|\chi|$), the probability of collisions increases with increasing number of records, especially when $L$ is fixed.  Therefore, we propose collision-included probabilistic counting (CIPC).  We calculate the expected number of collisions and add it to $2^k$ to improve the results.  Birthday paradox~\cite{Wagner02ageneralized} is adopted to estimate the expected number of collisions of multiplication hashing.

\begin{figure}
\begin{center}
\includegraphics[width=.5\textwidth]{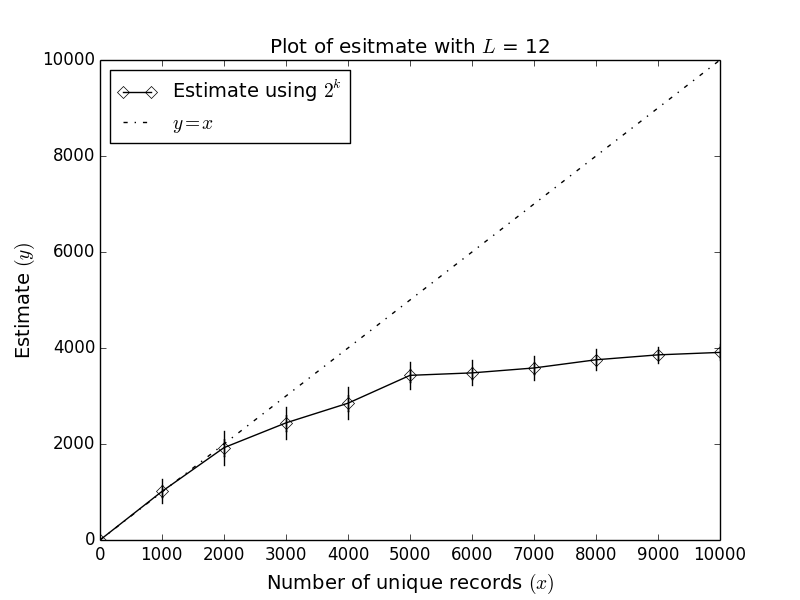}
\caption{Probabilistic Counting Estimate Deviation.}
\label{Fig:leveloff}
\end{center}
\end{figure}

\begin{theorem}
If the position of the rightmost zero in \textit{BITMAP} (ranks start at $0$) is used as an indicator of $\log_2 M$, under the assumption that the hash values are uniformly distributed, the expected value of $M$ including collisions is:
\begin{equation}\label{eq:cipc}
\tilde{M}_{CIPC}=\Big\lfloor\log_{1-\frac{1}{2^L}} \Big(1-2^{k-L}\Big)\Big\rfloor
\end{equation}
where $\tilde{M}_{CIPC}$ is the estimate of $M$ obtained using collision-included probabilistic counting (CIPC).
\end{theorem}  

\begin{proof}
Since the $M$ hash values are uniformly distributed over $[0,1,\cdots,2^L-1]$, according to birthday paradox~\cite{ProbHash}, the expected number of collisions is 
\begin{equation}\label{eq:collision}
C=M-2^L+2^L\Big(\frac{2^L-1}{2^L}\Big)^M
\end{equation}
where $C$ is the number of collisions.  Since $2^k$ is the estimate of $M$ without including the number of collisions, we have
\begin{equation}\label{eq:mcipc1}
\tilde{M}_{CIPC}\approx 2^k+C
\end{equation}
Substituting~\eqref{eq:collision} into~\eqref{eq:mcipc1} yields
\begin{equation}\label{eq:mcipc2}
\tilde{M}_{CIPC}\approx 2^k+ M -2^L+2^L\Big(\frac{2^L-1}{2^L}\Big)^M	
\end{equation}
Given that $\tilde{M}_{CIPC} = M$, solving~\eqref{eq:mcipc2} for $\tilde{M}_{CIPC}$ gives us
\begin{equation}
\tilde{M}_{CIPC}=\Big\lfloor\log_{1-\frac{1}{2^L}} \Big(1-2^{k-L}\Big)\Big\rfloor
\end{equation}
\end{proof}

Figure~\ref{Fig:flowchart} shows a flowchart of CIPC used to estimate the number of distinct system users.

\begin{figure}[!]
\begin{center}
\includegraphics[width=.5\textwidth]{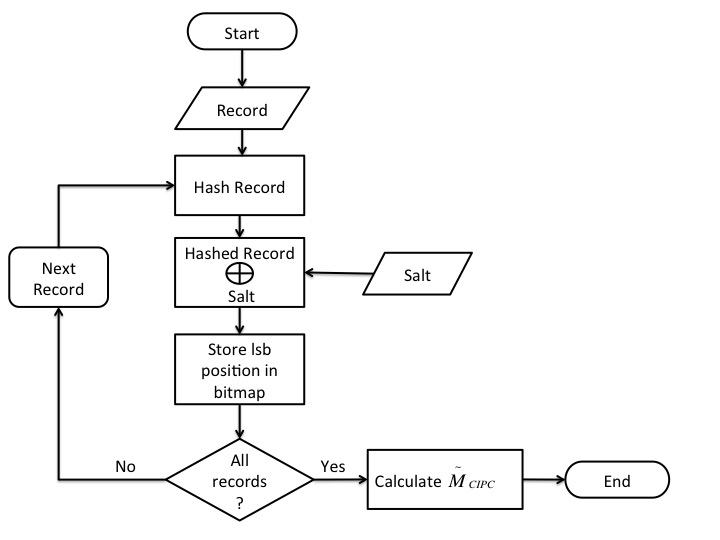}
\caption{Flow chart of the system.}
\label{Fig:flowchart}
\end{center}
\end{figure}

\begin{table*}[!]
\centering
\begin{tabular}{|L{1.2cm}|L{1.2cm}|L{2.2cm}|L{2.2cm}|L{1.2cm}|L{2.2cm}|L{2.2cm}|}
\hline
$M$ & $\tilde{M}_{PC}$  & $CI_{PC}$    & \textit{Percent Error} ($PE_{PC}$)  & $\tilde{M}_{CIPC}$ & $CI_{CIPC}$   & \textit{Percent Error} ($PE_{CIPC}$) \\ \hline

$1\times 10^3$    & 558   & (502,  614)     & 44.1506\% (min)         & 579    & (508, 649)     & 42.0687\% (min)           \\ \hline

$5\times 10^3$   & 2158  & (1856,  2460)   & 56.8314\%          & 2233   & (1862,  2603)   & 55.3373\%            \\ \hline

$1\times 10^4$   & 4236  & (3613,  4859)   & 57.6372\%          & 4348   & (3567,  5129)   & 56.514\%             \\ \hline

$2\times 10^4$   & 8869  & (7580, 10158)  & 55.6515\%          & 9236   & (7662,  10810)  & 53.816\%             \\ \hline

$3\times 10^4$    & 10590 & (10590,  10590) & 64.6977\%          & 11356  & (11356, 11356) & 62.1466\%            \\ \hline

$4\times 10^4$    & 18673 & (16211,  21134) & 53.3173\%          & 19649  & (16670,  22627) & 50.8772\%            \\ \hline

$5\times 10^4$   & 19769 & (17705,  21832) & 60.4614\%          & 20940  & (18351, 23529)  & 58.1189\%            \\ \hline

$6\times 10^4$   & 17651 & (11911,  23390) & 70.5814\% (max)        & 18283  & (11083,  25483) & 69.5277\% (max)           \\ \hline

$7\times 10^4$ & 32556 & (27636, 37476) & 53.4906\%          & 33476  & (27553,  39400) & 52.1758\%            \\ \hline

$8\times 10^4$   & 35673 & (30798,  40549) & 55.4076\%          & 37033  & (30917,  43150) & 53.7076\%            \\ \hline

$9\times 10^4$   & 34216 & (27734,  40697) & 61.98215\%         & 35205  & (27074,  43335) & 60.8833\%            \\ \hline

$1\times 10^5$  & 34125 & (28841,  39409) & 65.8744\%          & 35091  & (28462,  41719) & 64.9085\%            \\ \hline
\end{tabular}
\caption{Experimental result with $L= \lfloor log_2 M  \rfloor$}
\label{tbl:0}
\end{table*}

\begin{table*}[!]
\centering
\begin{tabular}{|L{1.2cm}|L{1.2cm}|L{2.2cm}|L{2.2cm}|L{1.2cm}|L{2.2cm}|L{2.2cm}|}
\hline
$M$ & $\tilde{M}_{PC}$  & $CI_{PC}$    & \textit{Percent Error} ($PE_{PC}$)  & $\tilde{M}_{CIPC}$ & $CI_{CIPC}$   & \textit{Percent Error} ($PE_{CIPC}$) \\ \hline

$1\times 10^3$    & 884    & (769, 1000)     & 11.5191\%          & 886    & (749, 1024)     & 11.3326\%            \\ \hline

$5\times 10^3$     & 3530   & (3097, 3962)    & 29.3954\%          & 3498   & (2971,4024)    & 30.033\%             \\ \hline

$1\times 10^4$    & 6586   & (5573, 7598)    & 34.1391\%         & 6504   & (5306,7702)    & 34.954\%             \\ \hline

$2\times 10^4$    & 14872  & (13048, 16695)  & 25.6398\%          & 15000  & (12802, 17197)  & 24.9991\%           \\ \hline

$3\times 10^4$    & 16878  & (14796,18961)  & 43.73699\%          & 17414  & (14873,19954)  & 41.9531\%            \\ \hline

$4\times 10^4$   & 28536  & (24627,32444)  & 28.6599\%          & 28344  & (23577, 33111)  & 29.1381\%           \\ \hline

$5\times 10^4$     & 31772  & (27465, 36078)  & 36.4558\%         & 32248  & (26922, 37574)  & 35.5024\%            \\ \hline

$6\times 10^4$   & 31772  & (26664,  36879)  & 47.0465\% (max)       & 32404  & (26193, 38615)  & 45.9927\% (max)           \\ \hline

$7\times 10^4$   & 89838  & (74745, 104931) & 28.3403\%          & 88414  & (71114, 105714) & 26.3064\%            \\ \hline

$8\times 10^4$    & 94758  & (80114, 109403) & 18.4484\%          & 93213  & (76323, 110103) & 16.5168\%           \\ \hline

$9\times 10^4$    & 100233 & (85382, 115083) & 11.3702\% (min)         & 99416  & (82157, 116675) & 10.4628\% (min) \\ \hline

$1\times 10^5$  & 68839  & (61495, 76183)  & 31.1605\%          & 71098  & (62014, 80182)  & 28.9014\%            \\ \hline
\end{tabular}
\caption{Experimental result with $L = \lfloor log_2 M  \rfloor + 1$}
\label{tbl:1}
\end{table*}

\begin{table*}[!]
\centering
\begin{tabular}{|L{1.2cm}|L{1.2cm}|L{2.2cm}|L{2.2cm}|L{1.2cm}|L{2.2cm}|L{2.2cm}|}
\hline
$M$ & $\tilde{M}_{PC}$  & $CI_{PC}$    & \textit{Percent Error} ($PE_{PC}$)  & $\tilde{M}_{CIPC}$ & $CI_{CIPC}$   & \textit{Percent Error} ($PE_{CIPC}$) \\ \hline
$1\times 10^3$   & 1109   & (919, 1299) & 10.9853\%          & 1034   &(822, 1245)   & 3.4\%                \\ \hline
$5\times 10^3$     & 5395   & (4446, 6344)  & 7.905\%  & 5248   & (4168, 6329)  & 4.9732\%             \\ \hline
$1\times 10^4$   & 10590  & (8972, 12208) &5.9068\%      & 10082 & (8227, 11936) & 0.8203\%            \\ \hline
$2\times 10^4$    & 22140  & (18757, 25522) &10.7\%       &  21374 & (17511, 25238) &6.8742\%         \\ \hline
$3\times 10^4$   & 28241  & (25109, 31374) & 14.2658\%           & 24377  &(21443, 27312) & 15.7221\% (max)           \\ \hline

$4\times 10^4$    & 44833  & (37089,  52578)   & 12.0847\%          & 38585  & (31495,  45676)  & 3.53508\%     \\ \hline

$5\times 10^4$    & 54718  & (46912,  62524)   & 9.437\%    & 47616  & (40421, 54812)  & 4.7661\%             \\ \hline

$6\times 10^4$   & 56898  & (49628, 64169)   &5.1683\% (min)        & 57107  & (48475,  65739)  & 4.8209\%          \\ \hline

$7\times 10^4$     & 77665  & (64776,  90553)   & 10.95\%          & 65704  & (53884,  77524)  & 6.1366\%             \\ \hline
$8\times 10^4$   & 92844  & (78001,  107688)  & 16.0562\% (max)          & 79811  & (66157,  93464)  & 0.2361\% (min) \\ \hline
$9\times 10^4$     & 101317 & (86236,  116398)  & 12.575\%          & 87580  & (73666,  101495) & 2.68812\%            \\ \hline
$1\times 10^5$   & 110613& (95856,  125371) & 10.6138\%        & 111035 & (93671,  128399) & 11.0354\%      \\ \hline
\end{tabular}
\caption{Experimental result with $L = \lfloor log_2 M  \rfloor + 2$}
\label{tbl:2}
\end{table*}

\subsection{Experiment Results}
To evaluate the proposed algorithm, a set of experiments are performed for different values of $M$ with different register sizes ,$L$, using both probabilistic counting and CIPC.  The same multiplication hash function was used in both algorithms.  

In each experiment, we randomly picked $M$ distinct integers that range between $1$ and $2\times 10^5$ and used them as the records.  For each value of $M$, the experiment was repeated for $50$ times using each algorithm.  We then took the average as the final estimate of $M$.

The results of the experiment are given in Table~\ref{tbl:0} and Table~\ref{tbl:1}.  Each table corresponds to a different register size ($L$).  Columns named ``$CI_{\ast}$'' contain 95\% confidence interval (CI) of the corresponding estimate.  For easy and direct comparison of the results, we also calculated the \textit{Percent Error} (PE) of each estimate using formula
\begin{equation}
\textit{Percent Error}=\Bigg|\frac{\text{Estimate}-\text{True value}}{\text{True value}}\Bigg|\times 100\%
\end{equation}
The lower the percent error (PE) is, the closer the estimate is to the true value.  Based on this rule, we compare the performance of CIPC with probabilistic counting in terms of estimation accuracy.
\begin{itemize}[leftmargin=*]
\item $L= \lfloor log_2 M \rfloor$: As shown in Table~\ref{tbl:0}, $PE_{CIPC}$ is smaller than $PE_{PC}$ for all tested $M$.  This means CIPC always produces a more accurate estimate than probabilistic counting in our experiments.  But the advantage is not significant (mostly 1-3\% lower in PE rate). 
\item $L= \lfloor log_2 M \rfloor+1$: As shown in Table~\ref{tbl:1}, CIPC gives about the same results as probabilistic counting.  No significant difference is observed among the estimates generated using the two algorithms.  But it is worth mentioning that both maximum $PE_{CIPC}$ and minimum $PE_{CIPC}$ are smaller than the corresponding parameters of probabilistic counting.
\item $L= \lfloor log_2 M \rfloor+2$: As shown in Table~\ref{tbl:2}, CIPC provides significantly superior estimate accuracy to probabilistic counting except for $M=3\times 10^4$ and $1\times 10^5$.  Half of the $PE_{PC}$ are several times bigger than the corresponding $PE_{CIPC}$.
\end{itemize}
Based on the results of our experiments, it is safe to say that CIPC outperforms probabilistic counting in general.


\section{Conclusion}\label{sec:conclusion}
To accurately estimate the number users of an anonymous system, we adopted probabilistic counting algorithm to develop collision included probabilistic counting (CIPC), which does not store identifiable information of system users. CIPC includes the hash collisions to the estimate which gives a more accurate estimate than the original probabilistic counting. Based on simulation results, we recommend using $L \approx \lfloor log_2\tilde{M} \rfloor + 2$ as the register size to maximize the results.

For the future work, we will explore the feasibility of probabilistic counting as an anonymity tool.  We will investigate the information leakage of the CIPC, which implies how much uncertainty the system gives and probability of a user being identified by an attacker.

\section{Acknowledgement}
This material is based upon work sponsored by the National
Science Foundation under Grants Nos. 1547164, 1544910
and 1643020. Any opinions, findings, and conclusions or
recommendations expressed in this material are those of the
authors and do not necessarily reflect the views of the National
Science Foundation.

\nocite{ex1,ex2}
\bibliographystyle{latex8}
\bibliography{latex8}

\end{document}